\documentclass[letterpaper, 10 pt, conference]{ieeeconf}  

\IEEEoverridecommandlockouts                              

\usepackage{amsmath,amsfonts,amssymb}
\usepackage{textcomp } 
\usepackage{mathrsfs}         

\newtheorem{theorem}{Theorem}

\newtheorem{corollary}{Corollary}
\newcommand{\eqdef}{:=}
\newcommand{\Fq}{\mathbb{F}_q}

\title{\LARGE \bf
Multi-Party Set Reconciliation Using Characteristic Polynomials
}

\author{Anudhyan Boral$^{1}$\thanks{$^{1}$Harvard University, School of Engineering and Applied Sciences.
Supported in part by NSF grant CCF-1320231.}
\and
Michael Mitzenmacher$^{2}$\thanks{$^{2}$Harvard University, School of Engineering and Applied Sciences.
Supported in part by NSF grants CCF-1320231, CNS-1228598, and IIS-0964473.
Part of this work was done while visiting Microsoft Research, New England.}
}

\begin{document}

\maketitle
\thispagestyle{empty}
\pagestyle{empty}


\begin{abstract}
In the standard set reconciliation problem, there are two parties
$A_1$ and $A_2$, each respectively holding a set of elements $S_1$ and
$S_2$. The goal is for both parties to obtain the union $S_1 \cup
S_2$.  In many distributed computing settings the sets may be large
but the set difference $|S_1-S_2|+|S_2-S_1|$ is small.  In these cases
one aims to achieve reconciliation efficiently in terms of
communication; ideally, the communication should depend on the size of
the set difference, and not on the size of the sets.

Recent work has considered  generalizations of the reconciliation problem
to multi-party settings, using a framework based on a specific type of
linear sketch called an Invertible Bloom Lookup Table.  Here, we
consider multi-party set reconciliation using the alternative framework of
characteristic polynomials, which have previously been used for
efficient pairwise set reconciliation protocols, and compare their
performance with Invertible Bloom Lookup Tables for these problems.
\end{abstract}

\section{Introduction}
In the standard theoretical framework for the {\em set reconciliation}
problem, two parties $A_1$ and $A_2$ each hold a set of keys
from a (large) universe $U$ with $|U| = m$, with the sets named
$S_1$ and $S_2$ respectively.  The goal is for both parties to
obtain $S_1 \cup S_2$.  Typically, set reconciliation is interesting
algorithmically when the sets are large but the set difference
$|S_1 - S_2| + |S_2 - S_1|$ is small; the goal is then to perform the
reconciliation efficiently with respect to the transmission size.
Ideally, the communication should depend on the size of the set
difference, and not on the size of the sets.

Recent work has examined the problem of extending set reconciliation
to multi-party settings \cite{IBLT,setreconc}.  This work examined the problem where three
or more parties $A_1,A_2,\ldots,A_N$ hold sets of keys
$S_1,S_2,\ldots,S_N$ respectively at various locations in a network,
and the goal is for all parties to obtain $\cup_i S_i$.  This could of
course be done by pairwise reconciliations, but more effective methods
are possible.  The multi-party set reconciliation problem is a natural
distributed computing problem.  For example, set reconciliation models
distributed loosely replicated databases, which for simplicity we
think of as simply holding a set of keys.  Such databases may be periodically
synchronized.  As we expect the number of differences among the databases
to be small compared to the database size, we would like reconciliation
schemes that scale with these differences.  We also would like to make use
of the network efficiently, ideally more efficiently than pairwise reconciliations.

Recent work has tackled this problem by extending Invertible Bloom
Lookup Tables (IBLTs), a hash-based data structure that, among other
uses, provide a natural solution to the two-party set reconciliation
problem \cite{IBLT}.  (See also \cite{IBF,setdifference}.)  The
extension shows that by performing operations on IBLTs in an
appropriate field, one can design protocols for multi-party set
reconciliation \cite{MitzPagh}.  Further, because IBLTs are linear
sketches, using IBLTs allows the use of network coding techniques to
improve the use of the network \cite{MitzPagh}.

In this work, we return to the classic solution for the 2-party
setting, which is based on characteristic polynomials and uses
techniques similar to those used for Reed-Solomon codes \cite{minsky2003set}.
The goal is to see whether and how much of the results for IBLTs in the
multi-party setting can be translated to similar results using these
alternative techniques.

We expect there to be a trade-off.  In the 2-party setting, where $d =
|S_1 - S_2| + |S_2 - S_1| = |S_1 \cup S_2|-|S_1 \cap S_2|$ is the set
difference, both techniques only require sending and receiving $O(d
\log m)$ bits of communication (as long as $d$, or an upper bound on
$d$ that is $O(d)$, is known), but using characteristic polynomials
generally requires a constant factor less communication than using
IBLTs, and is almost optimal in terms of communication.  In return,
using characteristic polynomials is much more computationally
intensive.  While IBLTs require only a linear number of operations to
recover all elements and $O(|S_i|)$ operations for each party to
compute the information to be transmitted (assuming suitable field
operations are $O(1)$ and hashing can be treated as a constant-time
operation), using characteristic polynomials requires almost $O(d^3)$ time to
recover all elements using standard Gaussian elimination and
$O(|S_i|d)$ operations to compute the information to be passed.  (As
discussed in \cite{setreconc}, and as we discuss further below,
theoretically faster algorithms are possible, but they are still
super-linear, and it appears that due to high constant factors they may
be unlikely to be useful in practice.)  Finally, IBLTs are randomized and
succeed with high probability, while using characteristic polynomials
is deterministic.

The main contribution of this work is to show that characteristic polynomials can, in a
suitable fashion, be extended to the multi-party case.
We follow the framework of the multi-party problem definition that is
considered in \cite{MitzPagh}, where the authors quantify the amount
by which the sets differ by the number of elements which belong to at
least one set but not all of them.  We call this quantity the
\textit{total set difference} $d$; here $d = |\cup_{i = 1}^N S_i
- \cap_{i = 1}^N S_i|$, which generalizes the set difference
for two parties.

We first show that there is a protocol in the {\em relay setting},
where each party is connected to an intermediate relay that can
compute and broadcast messages, in which each party and the relay send
a message of $O(d \log m)$ bits using an approach based on
characteristic polynomials. The communication is asymptotically
optimal information theoretically.

Building on this approach, we consider set reconciliation in an asynchronous network setting
using characteristic polynomials.  Each party is located at a distinct node
in a graph $G$ of size $N$, and in one round, only parties which
are on adjacent nodes can communicate with each other. Using recent
results from the network gossip literature \cite{giakkoupis},
we show that with each party sending (and receiving) at
most $O(d \log m)$ message bits in each round, it takes
$O(\phi^{-1}\log N)$ rounds for every node to obtain the union $S_\cup$
with high probability. Here $\phi$ is the \textit{conductance} of the graph $G$;
see, e.g., \cite{giakkoupis,shah} for more information on conductance.

Additionally, we show that with slight modifications, in both the
central relay and the network setting, our protocols can also support
recovery of owners of elements. That is, an agent $A_i$, after obtaining the union
of the sets $\cup_i S_i$, should also be able to recover an owner of the items
she does not own herself.  Specifically, at the end of the
protocol each party $A_i$ can not only obtain all the items in the set
$(\cup_j S_j) - S_i$, but $A_i$ can also obtain an original owner of each of
these items.

Although these results appear generally promising, we note they come
with significant limitations.  The intermediary nodes must do
significant work, essentially decoding sketches and recoding
information based on the decoding.  Less effort appears to be required by
intermediary nodes when using IBLTs, as the corresponding data
sketches are linear and can be combined using simple operations.
Hence, while our work shows that characteristic polynomials can be
used as a basis for multi-party reconciliation, we believe that further
simplification would be desirable.

\section{Background and Notation}

We work with the characteristic function of a set. For a set
$S \subset [m]$, and a prime $q$ greater than $m$, the characteristic function $f_S : \mathbb{F}_q
\mapsto \Fq$ is a polynomial defined as:
$$f_S(x) = \prod_{\alpha \in S} (x - \alpha)$$
where $\Fq$ is the prime finite field with $q$ elements.

A sketch $\sigma_d(p)$ of a polynomial $p : \Fq \mapsto \Fq$ is defined as a $(d + 1)$-tuple
of the evaluation of $p$ at $d + 1$ fixed points of $\Fq$. It is not important for us which $d + 1$
points we choose, but for concreteness let us fix those points to be $\{0,1,\ldots, d\}$. We note that
by Lagrange interpolation, it is possible to recover the coefficients of a degree $d$ polynomial $p$
from its sketch $\sigma_d(p)$. This is the key idea from Reed-Solomon codes that we exploit in our
protocol. By a sketch $\sigma_d(S)$ of a set $S \subset [m]$, we mean a sketch of its characteristic
function. Where there is no risk of confusion, we drop the subscript $d$ and refer to the
sketch of $S$ as $\sigma(S)$.

The standard approach for 2-party reconciliation using sketches of
this form is presented in \cite{minsky2003set}.  Treating keys as
numbers in a suitable field, $A_1$ considers the characteristic polynomial $f_{S_1}$ over a field
$\mathbb{F}_q$ for $q$ larger than $m$; and
similarly $A_2$ considers $f_{S_2}$.  Observe that in the rational
function $f_{S_1}/f_{S_2}$ the common terms cancel out, leaving a rational
function in $x$ where the sums of the degrees of the numerator and
denominator is the set difference, where the set difference is
defined as the quantity $|(S_1 - S_2) \cup (S_2 - S_1)|$. Assuming the set difference is at
most $d$, the rational function can be determined through
interpolation by evaluating the function at $d+1$ points, and then
factored.  Hence, if $A_1$ and $A_2$ send each other their respective
sketches $\sigma_d(S_1)$ and $\sigma_d(S_2)$, each party can compute
$f_{S_1}/f_{S_2}$ at $d+1$ points and thereby determine and reconcile the
values in $(S_1 - S_2) \cup (S_2 - S_1)$. The total number of bits sent in each
direction would be $(d+1) \lceil \log_2 q \rceil$.  Note that this
takes $O(d^3)$ operations using standard Gaussian elimination
techniques.  These ideas can be extended to use other codes, such as
BCH codes, with various computational trade-offs \cite{DORS}. Because
of the use of division to combine sketches, the sketches are not
``linear'' and do not naturally combine when used for three or more parties.
However, as we show, with a bit more work this limitation can be circumvented.

In the case of multiple parties, we define the total set difference of
the collection of sets $S_1, S_2, \ldots S_N$ to be the quantity
$|(\cup_{i = 1}^N S_i) - (\cap_{i = 1}^N S_i)|$. For convenience,
we assume in what follows that all parties know in
advance that the total set difference of the collection $\{S_1, \ldots
S_N\}$ does not exceed $d$.  Generally, in reconciliation settings,
there are multiple phases.  For example, in a first phase a bound on
$d$ is obtained, which is then used for reconciliation.
Alternatively, one takes an upper bound on $d$ that is suitable most
of the time, and then checks for successful reconciliation after the
algorithm completes using hashing methods.  See \cite{setdifference,minsky2003set}
for further discussion on this point.

\section{Multi-Party Reconciliation with a Central Relay}
We first describe a protocol where each of the $N$ parties $A_1 \ldots A_N$,
possessing sets $S_1 \ldots S_N$ respectively, communicate with a central relay in order to collectively
obtain the union of all the sets. We use the shorthands $S_\cup$ and $S_\cap$ to denote
$\cup_{i \in [N]} S_i$ and $\cap_{i \in [N]} S_i$. As mentioned we assume
that $|S_\cup - S_\cap| \leq d$, where $d$ is small compared to the number of elements
in the sets $S_i$.

The protocol is carried out as follows. Initially, each party $A_i$
computes their own sketch $\sigma_d(S_i)$ and sends the $O(d \log m)$ bits describing this
sketch to the relay. From these sketches, the relay computes the sketch of the union $S_\cup$. The relay
broadcasts the sketch of $S_\cup$ and from this sketch each party $A_i$ can retrieve the
elements of $S_\cup - S_i$.

\smallskip
\noindent {\bf Combining Sketches} We show how to combine the sketches of two sets to obtain the sketch of
their union.
In the following, we use the $\circ$ operator to denote coordinate-wise multiplication (in $\Fq$) of two sketches.
We also use $\circ^{-1}$ to denote coordinate-wise division of two sketches.
For two sets $S, T \subset [m]$ with a set difference of at most $d$, given
$\sigma_d(S)$, and $\sigma_d(T)$, we compute $\sigma_d(S \cup T)$ using the
following identity.
$$ \sigma_d(S \cup T) = \sigma_d(S) \circ \sigma_d(T - S)$$
The central relay can find $\sigma_d(T - S)$ from factoring the rational function $f_T/f_S$
and extracting the numerator, since, in its reduced form the rational function $f_T/f_S$ can be
written as
$$ \frac{f_T(x)}{f_S(x)} = \frac{ \prod_{\alpha \in T - S} (x - \alpha) }{
  \prod_{\alpha \in S - T} (x - \alpha)}.$$
Similarly the relay can find $\sigma_d(S - T)$.
Observe that the relay can recover the individual elements of $T - S$ and $S - T$
even though it does not have access to either of the sets $S$ and $T$ in its entirety.

By combining two sketches at a time, the central relay can obtain the sketch of $S_\cup =
\cup_{i \in [N]} S_i$ after $N - 1$ combinations. The relay then broadcasts the sketch
$\sigma_d(S_\cup)$ to each of the $N$ parties.

\smallskip
\noindent {\bf Note on Relay Output} The careful reader might notice that the relay can in fact just broadcast each
of the elements of $S_\cup - S_\cap$ to all the
parties. However, to maintain generality we work
with sketches throughout.  This approach allows us to generalize our
method to broader settings.

\smallskip
\noindent {\bf Distributed Computation at the Relay} If the relay has access to multiple
processors, she can perform the $N - 1$ combinations in a parallel manner. It is easy to see
that by combining two sketches at a time, using $N/2$ processors the relay needs to perform
$O(\log N)$ rounds of combinations.

\smallskip
\noindent {\bf Retrieving Missing Elements} Each $A_i$, having the sketch
$\sigma_d(S_\cup)$  and having computed already the sketch for her own set $S_i$, can retrieve the
individual elements she is missing from $S_\cup$.  The first step is to compute the sketch
of $S_\cup - S_i$ using component-wise division: $$ \sigma_d(S_\cup - S_i) = \sigma_d(S_\cup) \circ^{-1} \sigma_d(S_i).$$
From the sketch of $S_\cup - S_i$, agent $A_i$ can interpolate the characteristic function of
$S_\cup - S_i$, which is a polynomial of degree at most $d$. The actual elements of $S_\cup - S_i$
are then determined by factoring its characteristic function.

\smallskip
\noindent {\bf Tightness of Communication} The amount of communication in our protocol is
information-theoretically as succinct as possible in this setting, as the number of elements in the
universe grows, assuming that the parties have no prior
knowledge of the constituents of each others' sets other than the fact
that the total set difference is at most $d$. This is because, after fixing a particular
agent $A_i$'s set as $S_i$, with $|S_i| = n$, there are at least $\binom{m - n}{d}$
different subsets of elements that $A_i$ might be missing from the union. But if $m$ is much
larger compared to $n$ and $d$, then $\binom{m - n}{d} \approx m^d$. Thus, to specify the
missing elements we need at least $\log \binom{m - n}{d} \approx d\log m$ bits.

\smallskip
\noindent {\bf Time Complexity} We would like the computations performed by the
parties and the relay to be efficient. We aim for the time required to
encode to be linear in the number of set items, and the time to decode
to be polynomial in the total set difference $d$.  We assume that the basic
arithmetic operations in $\Fq$ take constant time. (This
may or may not be a reasonable assumption depending on the setting
and the size of $q$;  however, one can factor in a corresponding poly-logarithmic factor in $m$
as needed to handle the cost of operations in $\Fq$.)

Initially, each party $A_i$ must evaluate the polynomial corresponding to their set's characteristic
function at a pre-specified set of $d + 1$ points.  This can most straightforwardly be done by
the standard computation with $O(d|S_i|)$ operations, though for large total set differences it may be more efficient
(at least theoretically) to compute the coefficients of the
characteristic function and then evaluate the polynomial simultaneously
at the pre-specified set of $d + 1$ points.  (See, e.g., \cite{cantor} for possible algorithms.)
The computations performed by the relay include interpolation of a rational function where
the numerator and denominator have total degree $d$,
evaluation of a $d$-degree polynomial at $d$ points, and point-wise multiplication of
two sketches. Each of these computations can be done in $\tilde{O}(d)$ operations \cite{cantor}. (Here
$\tilde{O}$ hides polylog factors in $d$.)  The remaining
computation performed by the relay is factorization of a $d$-degree polynomial over $\Fq$;
the best theoretical algorithm of which we are aware is given by Kedlaya and Umans \cite{KU}, and requires
approximately $O(d^{1.5})$ operations;  other algorithms may be more suitable in practice.  (See also
\cite{KaltofenShoup}.)


We can state our results in the form of the following theorem:
\begin{theorem}
Given an upper bound $d$ on the size of the total set difference, $N$ parties using a relay
can reconcile their sets, from an universe of $m$ elements, using sketches of $d+1$ values in $\Fq$ (with $q > m$)
and with each party sending one sketch and the relay broadcasting a sketch.
Each sketch can be encoded in a message of $O(d \log m)$ bits.
The time for computation required by the $i$th party with set $S_i$ is
the time to evaluate their characteristic polynomial at $d+1$ pre-chosen points,
and the computation time required by the relay is dominated by
the time to factor a degree $d$ polynomial over $\Fq$ at most $O(N)$ times.
\end{theorem}

\smallskip
\noindent {\bf Recovering an Owner of a Missing Element} We describe modifications to our protocol
that would further enable each party to also retrieve an owner of each element she is missing from the union.
In some settings, this additional information be useful;  for example, there may be additional information
associated with a set element that may require contact between the parties to obtain or resolve.

The relay can maintain running sketches of the current union and intersection of the sets following
the framework that we have described. When a new sketch for set $S_i$ arrives, new running
values $S_\cup^*$ and $S_\cap^*$ can be computed from previous values $S_\cup'$ and $S_\cap'$ as follows:
$$ \sigma_d(S_\cup^*) = \sigma_d(S_\cup') \circ \sigma_d(S_i - S_\cup');$$
$$ \sigma_d(S_\cap^*) = \sigma_d(S_\cap') \circ^{-1} \sigma_d(S_\cap' - S_i).$$
Note that, as part of this process, by using $\sigma_d(S_i - S_\cup')$ to determine the
characteristic polynomial of $S_i - S_\cup'$, the relay can determine which new elements are
being brought into the union by each set as each sketch arrives.  (For the combination of the
first two sketches $S_i$ and $S_j$, both $\sigma_d(S_i - S_j)$ and $\sigma_d(S_j - S_i)$ will
need to be computed, as described above.)  As before, while the relay could broadcast this
information, we prefer to keep everything in the setting of sketches.

To produce a final sketch, the relay can then re-encode each element $\alpha$ in the final
$S_\cup'- S_\cap'$ by encoding the element as $(\alpha + m i_\alpha)$, where $i_\alpha$ is the
$0$-based index of the smallest indexed owner of the element $\alpha$. As we are working over
prime finite fields, the item value and owner ID are the remainder and the quotient
respectively from division of the encoded field element by $m$.
We now work over a larger field $\mathbb{F}_{q'}$ with $q' > mN$ and we denote the sketch of
any set $T$ as $\sigma_d^\omega(T)$ after re-encoding each element with an owner index.

The sketch $\sigma_d^\omega(S_\cup)$ is created by taking the point-wise multiplications of
the corresponding sketches of $S_\cap'$ and $S_\cup' - S_\cap'$.
$$ \sigma_d^\omega(S_\cup) = \sigma_d^\omega(S_\cap') \circ \sigma_d^\omega(S_\cup' - S_\cap'),$$
The final $S_\cup$ just equals the final computed $S_\cup'$;
note this corresponds to owner labels being set to $0$.

If the relay now sends $\sigma_d^\omega(S_\cup)$ to a party, with elements
marked by an owner, the party cannot use the previous cancellation
procedure as she does not know which owner her own elements were
assigned to. Therefore, the relay also sends a sketch of the
intersection, $\sigma_d^\omega(S_\cap)$. Each party thus receives the sketches
of both $S_\cup$ and $S_\cap$ and can hence obtain the sketch of $S_\cup - S_\cap$.
We can decipher all the elements $S_\cup - S_\cap$ along with one of their owners.

Alternatively, with the re-encoded values from the relay, the set
difference between the final $S_\cup$ and $S_i$ consists of at most
$2d$ elements, as $d$ elements in $S_i$ might have been re-encoded to
different values (that is, the same element but encoded to a different
owner).  Hence the relay could send a sketch for up to $2d$
differences, or $\sigma_{2d}(S_\cap)$.  If each party computes
$\sigma_{2d}(S_i)$, then after the relay broadcasts each party can
compute
$$\frac{f_{S_\cap}(x)}{f_{S_i}(x)}$$ by interpolation and thereby recover any missing
elements or elements that have been re-encoded.

\smallskip
\noindent {\bf Asynchronous Message Arrivals} We note that our protocols are robust enough to
be able to handle situations where the initial messages from the agents to the relay
arrive asynchronously. In this case, the relay can keep performing computations without
having to wait for all the $N$ messages to arrive. The encoding that maps from agent indices
to owner IDs remains fixed, but the relay has leeway in determining which owner ID is attached to a particular
item, when the item has multiple owners. In fact, the relay can define an arbitrary ordering
on the set of agents $\{1,2 \ldots N\}$ and use this ordering to choose the 'smallest' owners of items.
In particular, the relay can choose the order in which the messages from the agents arrive. This will enable
the relay to simplify computation by enabling her to perform the sketch combinations without having to wait
for subsequent messages to arrive.

\smallskip

We state our modifications to the protocol to enable recovery of the owners in the corollary below.

\begin{corollary}
Given an upper bound $d$ on the size of the total set difference, each of $N$ parties using a relay
can reconcile their sets and also obtain one owner of each missing element, using sketches of $d+1$
values in $\Fq$ (with $q > mN$), with each party sending one sketch
and the relay broadcasting a sketch - each a message of size $O(d \log (mN))$.
The time for computation required by the $i$th party with set $S_i$ is
the time to evaluate their characteristic polynomial at $d+1$ pre-chosen points,
and the computation time required by the relay is dominated by
the time to factor a degree $d$ polynomial over $\Fq$ at most $O(N)$ times.
\end{corollary}

\section{Reconciliation in the Network Setting}

In this section, we describe a protocol for multi-party set
reconciliation over a network. Using previous results on gossip spreading
techniques (also referred to generally as rumor spreading), we can
show that our protocol terminates in
$O(\phi^{-1} \log N)$ rounds of communication, where $\phi$ is
the conductance of the network.  Here again we are following the framework
of \cite{MitzPagh}, but replacing their use of IBLTs with sketches based on
characteristic polynomials.

In this setting, we assume that each of the $N$ nodes start with the
knowledge of only their own set, and aim to follow a gossip protocol
so that each of them obtains the union of all the $N$ sets within a
small number of rounds.  The $N$ parties are situated at $N$ different
nodes of a graph $G$ and only adjacent nodes can communicate with each
other directly.  To be clear the graph $G$ may have more than $N$ nodes,
as there will generally be nodes that pass messages that are not
parties with information.

In the case where one party has a piece of information to distribute
to all other parties, it is known that
the standard {\sc PUSH}-{\sc PULL} protocol for ``rumor spreading''
will distribute that information to all the nodes of the graph
within $O(\phi^{-1} \log N)$ rounds high probability \cite{giakkoupis}.
(The {\sc PUSH}-{\sc PULL} protocol works as follows:  in each round,
every informed node that knows the rumor to be spread chooses a random
neighbor and sends it to the neighbor;  every uninformed node that does
not know the rumor contacts a neighbor in an attempt to get the rumor.)
For more on rumor spreading,
see also for example \cite{chierichettialmost,DebMedard,haeupler,shah}).

Here, we explain how the approach used by the relay described
previously allows us to use the standard {\sc PUSH}-{\sc PULL}
protocol for reconciliation.  (The general approach will also apply to allow
us to use other rumor spreading protocols for reconciliation.)

We show that the previous protocol described for relays without the
owner information carries over to the network setting using the {\sc
PUSH}-{\sc PULL} protocol. In a particular round, a node $v$ would
possess the sketch of the union of the sets belonging to a
sub-collection of the agents, say $C \subset [N]$. Let us denote this
set as $S_C$ ($\eqdef \cup_{i \in C} S_i$) where $C$ can be an arbitrary sub-collection of $[N]$.
We divide each round into two sub-rounds.  In the first sub-round, each vertex
pushes whatever information it has to a random neighbor.  In the
second sub-round, each vertex pulls whatever information it can obtain
from a random neighbor.  (Here, the sets are the rumors, and each passed sketch encodes the information about
the union of the sets obtained from all previous rounds. At the end,
we use a union bound over all possible rumors.)

If the node $v$ receives a sketch of $S_D$ from an adjacent neighbor,
for some $D \subset [N]$, then she can update her known sub-collection
to be $C \cup D$ and obtain the sketch of $S_{C \cup D}$. This
procedure of combining the two sketches is the same as the one
performed by the relay in the central relay setting. Note that, thus
far, the protocol we are considering does not carry information which
would enable $v$ to know \textit{which} sub-collection $C$ corresponds
to the set that she is holding.

Using the known bound on the {\sc PUSH}-{\sc PULL} protocol have the following theorem.

\begin{theorem}
Given an upper bound $d$ on the size of the total set difference, $N$ parties each possessing
sets from a universe of size $m$ communicating over a graph $G$ can reconcile their sets
using sketches of $d+1$ values in $\Fq$ (such that $q > m$), with each party sending
one sketch per sub-round (an $O(d\log m)$ bit message) using the {\sc PUSH}-{\sc PULL} randomized gossip
protocol, in $O(\phi^{-1} \log N)$ rounds with high probability.
\end{theorem}

\begin{proof}
We choose a suitable number of rounds $L = O(\phi^{-1} \log n)$ based on the
desired high probability bound that allows
$N$ parallel versions of the single-message gossip protocol to
successfully complete with high probability, as guaranteed by Theorem 1.1
of \cite{giakkoupis}.

Note that, for any single set, the sketch corresponding to that set
behaves just as though it was acting as part of the
single-message protocol;  the fact that other sketches may have been combined
into a shared sketch does not make any difference from the point of view
of the single set under consideration.
Hence, we can treat this as multiple single-message
problems running in parallel, and apply a union bound on the failure
probability.  (See \cite{MitzPagh} for a more extensive discussion.)

Hence, after $L$ rounds, with high probability all $N$ parties obtain sketches for all of the $N$ sets,
and hence all parties have the necessary information for reconciliation.
\end{proof}

\smallskip
\noindent {\bf Gossip Protocol with Owner Information} We now turn to extending the above
protocol to allow owner information to be be determined as well.
The primary difficulty in adapting our protocol with owner information
from the relay setting is that we might have the same item identified with different owners, and they
might be mistaken for different items. However, if we carry along the sketch of the intersection
as well, during the combination of two sets each party can extract and re-encode the items belonging to
at least someone but not everyone in the sub-collection.

Instead of $0$-indexed identifiers of the nodes, we use labels $\{1,\ldots N\}$ for the agents
for the purpose of owner identification. We work over of a prime field $\Fq$ with $q$ a prime which is
at least $m(N + 1)$. For a given
sub-collection $C \subset [N]$, we maintain the sketch of the union and the intersection of the original
sets corresponding to the sub-collection $C$. We denote these sets by the shorthands $S_{\cup C}$ and
$S_{\cap C}$ respectively.  We call the minimum indexed member of $C$
the \textit{leader} of $C$, and denote her index by $l_C$.

When we consider the set $S_{\cup C}$, for each element of $S_{\cup C}-S_{\cap C}$, we also attach the index of the minimum-indexed owner
from among agents in $C$, following the approach used to associate an agent with an item in the relay case.
For elements in $S_{\cap C}$, we store the elements with a dummy owner value
of $0$. (Otherwise, the elements of $S_{\cap C}$ would have had $l_C$ as the minimum-indexed owner.)
To distinguish this encoding from the previous protocols, we denote the
this sketch of a set $S$ as $\hat{\sigma}(S)$.
In our protocol, at each node $v$ we maintain the
running tuple $\tau(C) \eqdef (\hat{\sigma}(S_{\cup C}), \hat{\sigma}(S_{\cap C}), l_C)$, where
$C$ is the collection of agents whom the node $v$ has made contact with, either directly or indirectly.
(That is, it is the collection of agents whose original sketches have reached $v$, albeit perhaps combined with other sketches along the way.)
Again, note that in both the sketches $\hat{\sigma}(S_{\cup C})$ and $\hat{\sigma}(S_{\cap C})$, the elements which also appear in
$S_{\cap C}$ are encoded with the owner-id of $0$ instead of $l_C$.

\smallskip
\noindent {\bf Combining Sub-Collection Sketches}
When a node $v$ receives information
about another sub-collection $D \subset [N]$ from a random neighbor, it combines
this with its own tuple for  $C \subset [N]$
to obtain the tuple $\tau(E) = (\hat{\sigma}(S_{\cup E}), \hat{\sigma}(S_{\cap E}), l_E)$,  where
$E \eqdef C \cup D$. Note that it is not necessary that $C$ and $D$ be disjoint sub-collections.

We now describe how to combine $\tau(C)$ and $\tau(D)$
to compute $\tau(C \cup D) = \tau(E)$.  We are given the sketches of $S_{\cup C}, S_{\cap C}, S_{\cup D},$ and
$S_{\cap D}$, as well as the minimum indices $l_C$ and $l_D$.
We assume without loss of generality that $l_C \leq l_D$, so $l_E = l_C$.  To compute
the sketches $\hat{\sigma}(S_{\cup E})$ and $\hat{\sigma}(S_{\cap E})$, the key idea is to extract the
sketch of $S_{\cap E}$ and then explicitly recover and re-encode all the elements in
both $S_C$ and $S_D$ which do not belong to agent $A_{l_E}$. We describe the procedure
in three parts: (a) Computing the sketch of $S_{\cap E}$. (b) Extracting the elements in $S_{\cup E} - S_{\cap E}$
and re-encoding their owner information. (c) Computing the sketch of $S_{\cup E}$.

(a) First, we show how to compute the sketch of $S_{\cap E}$. Note that the items of $S_{\cap C}$ have
$l_C$ as their minimum-indexed owner, hence they will be encoded with an owner-ID of $0$ in both the
sketches $\hat{\sigma}(S_{\cap C})$ and $\hat{\sigma}(S_{\cup C})$. An analogous statement holds for the
sub-collection $D$. As both $\hat{\sigma}(S_{\cap C})$ and $\hat{\sigma}(S_{\cap D})$ have all their
owner-IDs equal to $0$, we can treat these as sets of items without owner information. Using the identity
$$ \frac{f_T(x)}{f_S(x)} = \frac{ \prod_{\alpha \in T - S} (x - \alpha) }{
  \prod_{\alpha \in S - T} (x - \alpha)}.$$
we can compute $\hat{\sigma}(S_{\cap C} - S_{\cap D})$. Now, we can perform point-wise division
of the sketches to compute $\hat{\sigma}(S_{\cap C} \cap S_{\cap D})$ using,
$$ \hat{\sigma}(S_{\cap C} \cap S_{\cap D}) = \hat{\sigma}(S_{\cap C}) \circ^{-1} \hat{\sigma}(S_{\cap C} - S_{\cap D}) $$
This gives us the sketch of $S_{\cap E} = S_{\cap C} \cap S_{\cap D}$.

(b) Next, we extract the individual items of $S_{\cup C} - S_{\cap E}$ along with their owners. Observe that all
items in $S_{\cap E}$ also belong to the leader of $C$, namely $A_{l_C}$. In both $\hat{\sigma}(S_{\cup C})$ and
$\hat{\sigma}(S_{\cap E})$ these items have the same owner ID of $0$. Hence, we can take advantage of the identity
$$ \frac{f_{S_{\cup C}}(x)}{f_{S_{\cap E}}(x)} = \prod_{\alpha \in S_{\cup C} - S_{\cap E}} (x - \alpha - mi_\alpha) $$
where $i_\alpha$ is the smallest-indexed owner of item $\alpha$ among agents in the sub-collection $C$. Thus, we
have explicitly obtained the items in $S_{\cup C} - S_{\cap E}$ along with their minimum indexed owner.
Note that these are the items
which belong to some agent in the collection $C$, but not to all agents in the collection $E = C \cup D$.

Using a similar procedure on sub-collection $D$ instead of $C$, we obtain analogously the items
in $S_{\cup D} - S_{\cap E}$; encoded with their minimum-indexed owners among agents in $D$. At this point, if $l_D \neq l_C$,
we modify the items in $S_{\cup D} - S_{\cap E}$ with owner ID $0$ to instead have owner ID $l_D$. For items which appear both in
$S_{\cup C} - S_{\cap E}$ and $S_{\cup D} - S_{\cap E}$, we can now compute their minimum-indexed owner ID from among
agents in $E$  by choosing the minimum of the two owners that was obtained from the two sketches of $S_{\cup C} -
S_{\cap E}$ and $S_{\cup D} - S_{\cap E}$.

(c) In our final sketch $\hat{\sigma}(S_{\cup E})$, the items which also appear in $\hat{\sigma}(S_{\cap E})$ would have
agent $A_{l_E}$ as their minimum-indexed owner; hence by definition of $\hat{\sigma}$ these items appear in both the sketches
with an owner ID of $0$. This suggests that we can build the sketch $\hat{\sigma}(S_{\cup E})$ by using the already-computed
sketch $\hat{\sigma}(S_{\cap E})$ as a starting point. In fact, $S_{\cup E} = S_{\cap E} \sqcup
((S_{\cup C} - S_{\cap E}) \cup (S_{\cup D} - S_{\cap E}))$, where $\sqcup$ denotes disjoint union. Recall that we have the items
$\alpha$ of $(S_{\cup C} - S_{\cap E}) \cup (S_{\cup D} - S_{\cap E})$ explicitly, along with their minimum-indexed owner
ID (say $i'_\alpha$) from among agents in $E$. Thus we have
$$f_{S_{\cup E}}(x) = f_{S_{\cap E}}(x) . \prod_{\alpha \in (S_{\cup C} \cup S_{\cup D}) - S_{\cap E}}
(x - \alpha - m i'_\alpha).$$

We can therefore compute $\hat{\sigma}(S_{\cup E})$ by constructing the sketch of
the characteristic function of $(S_{\cup C} \cup S_{\cup D}) - S_{\cap E}$ and
point-wise multiplying it with the sketch of $S_{\cap E}$.

We summarize our result in the following corollary.
\begin{corollary}
Given an upper bound $d$ on the size of the total set difference, $N$ parties each possessing
sets from a universe of size $m$ communicating over a graph $G$ can reconcile their sets while recovering
an owner for each missing element, using sketches of $d+1$ values in $\Fq$ (such that $q > m(N + 1)$), with each party sending
at most two sketches per sub-round ($O(d\log(mN))$ bit messages) using the {\sc PUSH}-{\sc PULL} randomized gossip
protocol, in $O(\phi^{-1} \log N)$ rounds with high probability.
\end{corollary}

\section{Conclusion}

We had found that while the characteristic polynomial approach to set
reconciliation has been known for some time, the issue of considering
generalizations to multi-party settings had never apparently been
suggested. Linear sketches based on Invertible Bloom Lookup Tables
allow fairly straightforward multi-party reconciliation protocols.
In this work, we show that using characteristic polynomials can as well,
albeit perhaps somewhat less naturally and with more computation requirements.

A possible future direction is to improve the computation time requirements. Currently,
the primitives that we use for finite field arithmetic are not especially attuned to our
needs. Is it possible to take advantage of properties of finite fields to
to enable more efficient manipulation of sketches?

It would also be interesting to investigate if this approach could be simplified further,
as characteristic polynomials provide one of the simplest and most natural
frameworks for reconciliation problems.

\end{document}